\documentclass[conference]{IEEEtran}
\IEEEoverridecommandlockouts

\usepackage[ruled]{algorithm2e}
\usepackage{algorithmic}

\usepackage{amsmath}
\usepackage{amssymb}
\usepackage{amsfonts}
\usepackage{nccmath}
\newtheorem{proposition}{Proposition}

\usepackage{mathtools}
\usepackage{graphicx}
\graphicspath{ {./images/} }
\usepackage{caption}
\usepackage{subcaption}
\usepackage{cleveref}
\usepackage{color}
\usepackage{xcolor}
\usepackage{float}

\usepackage{cleveref}
\crefformat{equation}{(#1)}
\crefname{algorithm}{algorithm}{algorithms}
\usepackage{cite}
\usepackage{bbm}
\usepackage{bm}
\usepackage{comment}
\usepackage{arcs}

\usepackage{resizegather}

\usepackage{acronym}
\newacro{lqr}[LQR]{Linear Quadratic Regulator}
\newacro{slqr}[SLQR]{Structured Linear Quadratic Regulators}
\newacro{olqr}[OLQR]{Output-feedback Linear Quadratic Regulators}
\newacro{lqg}[LQG]{Linear Quadratic Gaussian}
\newacro{dare}[DARE]{Discrete-time Algebraic Riccati Equation}
\newacro{ouralgo}[RNPO]{Riemannian Newton-type Policy Optimization}
\newacro{po}[PO]{Policy Optimization}
\newacro{pg}[PG]{Projected Gradient}
\newacro{gf}[GF]{Gradient Flow}
\newacro{gd}[GD]{Gradient Descent}
\newacro{sgd}[SGD]{Stochastic Gradient Descent}
\newacro{sde}[SDE]{Stochastic Differential Equation}
\newacro{PBH}[PBH]{Popov-Belevitch-Hautus}

\usepackage{cite}
\usepackage{textcomp}
\def\BibTeX{{\rm B\kern-.05em{\sc i\kern-.025em b}\kern-.08em
		T\kern-.1667em\lower.7ex\hbox{E}\kern-.125emX}}

\usepackage[inline]{enumitem}

\def\bR{{\mathbb{R}}}

\def\lyap{{\,\mathbb{L}}}

\newcommand\lyaptrace{{Lyapunov-trace~}}
\DeclareMathOperator{\diff}{\mathrm{d}}

\def\lyaptrace{{Lyapunov-trace~}}

\def\Amatrices{{\mathbb{R}^{n\times n}}}
\def\Astable{{\mathcal{M}}}

\newcommand{\tr}[1]{\ensuremath{\mathrm{tr}\left[ #1 \right]}}

\newcommand{\E}[2]{\ensuremath{\mathbb{E}_{#1}\left[ #2 \right]}}

\newcommand{\boxend}{\hfill \ensuremath{\Box}}

\DeclareMathOperator{\lambdamax}{\text{$\overline{\lambda}$}}
\DeclareMathOperator{\lambdamin}{\text{$\underline{\lambda}$}}

\makeatletter
\newcommand{\algorithmfootnote}[2][\footnotesize]{%
  \let\old@algocf@finish\@algocf@finish%
  \def\@algocf@finish{\old@algocf@finish%
    \leavevmode\rlap{\begin{minipage}{\linewidth}
    #1#2
    \end{minipage}}%
  }%
}
\makeatother

\title{\LARGE Learning Kalman Policy for Singular Unknown Covariances via Riemannian Regularization}
\author{Larsen Bier, \emph{Student Member, IEEE} and Shahriar Talebi, \emph{Member, IEEE}
\thanks{The authors are with the University of California, Los Angeles, CA, USA. Emails: {\em\small larsenbier@g.ucla.edu} and {\em\small s.talebi@ucla.edu}.}
}

\begin{document}

\maketitle
\begin{abstract}
Kalman filtering is a cornerstone of estimation theory, yet learning the optimal filter under unknown and potentially singular noise covariances remains a fundamental challenge. In this paper, we revisit this problem through the lens of control–estimation duality and data-driven policy optimization, formulating the learning of the steady-state Kalman gain as a stochastic policy optimization problem directly from measurement data.
Our key contribution is a Riemannian regularization that reshapes the optimization landscape, restoring structural properties such as coercivity and gradient dominance. This geometric perspective enables the effective use of first-order methods under significantly relaxed conditions, including unknown and rank-deficient noise covariances. Building on this framework, we develop a computationally efficient algorithm with a data-driven gradient oracle, enabling scalable stochastic implementations.
We further establish non-asymptotic convergence and error guarantees enabled by the Riemannian regularization, quantifying the impact of bias and variance in gradient estimates and demonstrating favorable scaling with problem dimension. Numerical results corroborate the effectiveness of the proposed approach and robustness to the choice of stepsize in challenging singular estimation regimes.
\end{abstract}

\section{Introduction}

Kalman filtering—recognized as the minimum mean squared error estimator for linear Gaussian systems—has a long-standing role in estimation theory since its introduction in \cite{kalman1960new}. Its extensions have been extensively studied, particularly within the framework of adaptive Kalman filtering \cite{mehra1970identification,mehra1972approaches,carew1973identification,belanger1974estimation,myers1976adaptive,tajima1978estimation}. More recently, \cite{zhang2020identification} categorizes existing methods into four principal approaches: Bayesian inference \cite{magill1965optimal,hilborn1969optimal,matisko2010noise}, maximum likelihood estimation \cite{kashyap1970maximum,shumway1982approach}, covariance matching \cite{myers1976adaptive}, and innovation correlation techniques \cite{mehra1970identification,carew1973identification}.
Among these, Bayesian and maximum likelihood approaches are often associated with significant computational overhead, while covariance matching methods can suffer from practical bias issues. Consequently, innovation correlation–based methods have gained greater prominence and have been the focus of more recent developments \cite{odelson2006new,aakesson2008generalized,dunik2009methods}. Despite their popularity, these methods rely heavily on underlying statistical assumptions and, importantly, lack non-asymptotic performance guarantees.

The duality between control and estimation establishes a fundamental connection between two core synthesis problems in systems theory \cite{kalman1960new,kalman1960general,pearson1966duality}. This relationship has long enabled the transfer of both theoretical insights and computational methodologies across these domains.
On the optimal control side, recent years have witnessed significant progress in data-driven synthesis techniques. In particular, first-order optimization methods have been successfully applied to state-feedback \ac{lqr} problems \cite{bu2019lqr,bu2020policy}. This policy optimization viewpoint has proven especially powerful, as the \ac{lqr} objective is known to satisfy a \emph{gradient dominance} property \cite{fazel2018global}. 
As a result, despite the inherent non-convexity of the problem when parameterized directly by the control policy, first-order methods can be employed with guarantees of global convergence.
Building on this foundation, first-order \ac{po} methods have been extended to a range of \ac{lqr}-type settings, including output-feedback \ac{lqr} (\ac{olqr}) \cite{fatkhullin2020optimizing,krasiler2024output}, model-free formulations \cite{mohammadi2021linear}, risk-constrained variants \cite{zhao2021global,talebi2024ergodic}, and \ac{lqg} problems \cite{Tang2021analysis}. More recently, geometric approaches based on Riemannian optimization have been developed, including optimization on submanifolds and extensions with ergodic-risk constraints \cite{talebi2022policy,talebi2024ergodic}; see \cite{talebi2024policy} for a comprehensive overview.

The extension of policy optimization (\ac{po}) methodologies to the estimation domain was initiated in \cite{talebi2022duality,talebi2023data}, where the optimal Kalman gain is learned in the presence of \emph{unknown noise covariances}. In this line of work, the problem is formulated as a stochastic policy optimization task minimizing output prediction error, linking data-driven optimal control with its dual, optimal filtering. Convergence guarantees are established for stochastic gradient methods under biased gradients and stability constraints, with bias–variance bounds scaling logarithmically in system dimension and trajectory length affecting only the bias.
Following this direction, \cite{belabbas2025interpretable} considers an slightly different MSE cost based on the steady-state innovation (prediction error), yielding a gradient that admits an interpretable decomposition as the product of an observability Gramian and a term capturing violation of the orthogonality principle.
In a related but distinct direction grounded in invariant ellipsoid theory, \cite{Khlebnikov2023comparison} develops an optimization-based filtering framework for systems subject to bounded disturbances, employing a Euclidean $\ell_2$-regularized gradient method. Despite these advances, there remains a notable gap in learning ill-conditioned estimation problems, particularly through non-Euclidean regularization in settings involving singular noise covariance structures.

In this paper, we study the estimation problem for linear systems with known dynamics and observation models, but with \emph{unknown and singular} process and measurement noise covariances. The objective is to learn the optimal steady-state Kalman gain from training data comprising independent realizations of the observation process. Building on recent developments in Riemannian policy \cite{talebi_riemannian_2022,talebi2022policy} optimization data-driven estimation \cite{talebi2022duality,talebi2023data}, we develop a first-order optimization method tailored to ill-conditioned estimation settings.
Our approach revisits classical estimation through the perspective of geometric regularization and control–estimation duality. In particular, we introduce a Riemannian regularization inspired by the Riemannian metric introduced in \cite{talebi2022policy} that restores key structural properties such as coercivity and gradient dominance of the cost over sublevel sets in the case of singular matrix parameters. This enables the effective application of first-order policy iteration methods under significantly relaxed assumptions, notably allowing for unknown and rank-deficient noise covariances.

Our contributions are fourfold. First, we formulate the estimation task as a policy optimization problem (\S\ref{sec:problem}), and introduce a Riemannian regularization that improves the conditioning and geometric structure of the problem (\S\ref{sec:geometry-algorithm}). Second, we develop a direct policy optimization framework for learning the optimal Kalman gain in the presence of unknown and singular noise covariances (\S\ref{sec:the-algorithm}). Third, we construct a data-driven gradient oracle from measurement sequences, which enables a stochastic implementation of the proposed method (\S\ref{sec:data-driven-oracle}). Fourth, we establish non-asymptotic error guarantees while preserving computational efficiency (\S\ref{sec:convergence}). 
Finally, we present numerical examples in \S\ref{sec:simulation} and conclude the paper in \S\ref{sec:conclusion}.

 \section{Background and Problem Formulation}\label{sec:problem}
Consider the stochastic difference equation,
\begin{subequations} \label{eqn:sysdyn}
\begin{align}
    x(t+1) =& A x(t) + \xi(t),\\
    y(t) =& H x(t) + \omega(t),
\end{align}
\end{subequations}
where $x(t) \in \mathbb R^n$ is the state of the system, $y(t)\in \mathbb R^m$ is the observation, and $\{\xi(t)\}_{t\in \mathbb Z}$ and $\{\omega(t)\}_{t\in \mathbb Z}$ are the uncorrelated zero-mean process and measurement noise vectors, respectively, with the following covariances,
\[\E{}{\xi(t)\xi^\intercal(t)} = Q \in \bR^{n\times n}, \quad \E{}{\omega(t)\omega^\intercal(t)} = R \in \bR^{m\times m},\]
for some positive semi-definite matrices $Q, R \succcurlyeq 0$.  Let $m_{0}$ and $P_0 \succcurlyeq 0$ denote the mean and covariance of the initial condition $x_0$. 
Also, let us fix a time horizon $T>0$ and define an estimation policy, denoted by $\mathcal L$, as a map that takes a history of the observation signal $\mathcal Y_T=\{y(0),y(1),\ldots,y(T-1)\}$ as an input and outputs an estimate of the state $x(T)$, denoted by $\hat{x}_{\mathcal L}(T)$. 

We make the following assumptions in our problem setup:
\begin{enumerate*}
    \item The system parameters $A$ and $H$ are known.
    \item The process and the measurement noise covariance matrices, $Q$ and $R$, are \emph{not available} and \emph{may be singular}.
    \item We have access to a training data-set that consists of independent realizations of the observation signal $\{y(t)\}_{t=0}^T$. However, ground-truth measurements of $x(T)$ is \emph{not} available.
\end{enumerate*}
This setting arises in applications such as active aero-elastic control of aircraft, where systems are well understood and admit approximate or reduced-order models, yet in deployment are subject to unmodeled dynamics, disturbances, and other uncertainties captured through process and measurement noise. Allowing the covariances Q and R to be rank-deficient enables modeling more structured disturbances, but also leads to an ill-posed estimation problem that complicates learning.

Ideally, one would seek an estimation policy $\mathcal L$ that minimizes the mean-squared state estimation error
\(
   \E{}{\|x(T)-\hat{x}_{\mathcal L}(T)\|^2},
\)
but this objective is infeasible since the true state $x(T)$ is not observable.
As an alternative, we consider a surrogate objective that minimizes the mean-squared error in predicting the observation $y(T)$ via $\hat{y}_{\mathcal L}(T)=H\hat{x}_{\mathcal L}(T)$. This constitutes a prediction problem, as $\hat{x}_{\mathcal L}(T)$ depends only on observations up to time $T-1$. The resulting optimization problem is to find $\mathcal L$ minimizing the mean-squared prediction error,
\begin{equation}\label{eq:min-pred-error-P}
    \min_{\hat{y}_{\mathcal L}(T) \in \sigma\{y(t)\}_{t=0}^T}  J^{\text{est}}_T(\mathcal L) \coloneqq \E{}{\|y(T)-\hat{y}_{\mathcal L}(T) \|^2},
\end{equation}
where $\sigma\{y(t)\}_{t=0}^T$ denotes the $\sigma$-algebra generated by past measurements up to current time $T$. 

\subsection{Kalman Policy  Parameterization}

Indeed, when $Q$ and $R$ are known, the solution is given by the celebrated Kalman filter algorithm~\cite{kalman1960new},\cite[Theorem 6.42]{kwakernaak1969linear}. 
For the case with $R\succ 0$ or $H Q H^\intercal \succ 0$ (or both), the unique optimal filtering policy iteratively updates the estimate $\hat{x}(t)$ according to \cite{Tse1970optimal} 
\begin{equation}\label{eq:KF-mean}
    \hat{x}(t+1) = A\hat{x}(t) + L(t)(y(t) - H \hat{x}(t)),~ \hat{x}(0) = m_{0},
\end{equation}
where $L(t):=AX(t)H^\intercal(HX(t)H^\intercal + R)^{-1}$ is the Kalman gain, and $X(t):=\mathbb E[(x(t) - \hat{x}(t))(x(t) - \hat{x}(t))^\intercal]$ is the error covariance matrix that satisfies the Ricatti equation,
\begin{gather*}
    X(t+1) = (A-L(t)H)X(t)(A-L(t)H)^\intercal + Q + L(t) R L(t)^\intercal 
\end{gather*}
with $X(t_0) = X_0$. It is known that $X(t)$ converges to a steady-state value $X^*$ when the pair $(A,H)$ is observable and the pair $(A,Q^{\frac{1}{2}})$ is controllable~\cite{kwakernaak1969linear,lewis1986optimal}. In such a case, the gain converges to $L^* :=AX^* H^\intercal(HX^* H^\intercal + R)^{-1}$, the so-called steady-state Kalman gain. For relatively large horizons $T$, it is a common practice to evaluate the steady-state Kalman gain $L^*$ offline and use it, instead of $L(t)$, to update the estimate in real-time.

We consider restriction of the estimation policies $\mathcal L$ to the Kalman filter realized with a \emph{constant gain} $L$. In particular, we define the estimate $\hat{x}_L(T)$ at time $T$ through \eqref{eq:KF-mean} with a the constant gain $L$ replacing $L(t)$. 

\begin{align}\label{eq:estimate-x-L}\textstyle
    \hat{x}_L(T) = A_L^Tm_0 + \sum_{t=0}^{T-1}A_L^{T-t-1} L y(t),
\end{align}
where $A_L\coloneqq A-LH$.
Note that this estimate does not require knowledge of the matrices $Q$ or $R$. By considering $\hat{y}_L(T):=H\hat{x}_L(T)$, the problem is now finding the optimal gain $L$ that minimizes the mean-squared prediction error
\begin{equation}\label{eq:mse-x-L}
  J^{\text{est}}_T(L):=\E{}{\|y(T)-\hat{y}_{L}(T)\|^2}.
\end{equation}

For the case of \emph{positive definite} noise covariances $Q$ and $R$, this problem has been recently studied in \cite{talebi2022duality, talebi2023data} and \acf{sgd}-type algorithm is proposed for guaranteed learning of the globally optimal Kalman gain. 
However, there is no guarantee these algorithms work for the ill-posed problems with rank deficient or singular noise covariances, simply because the pillar conditions of coercivity and gradient dominance fail to hold.

\textbf{Notation.} By $\Astable \coloneqq \left\{ A \in \Amatrices \;|\; \rho(A) <1 \right\}$, we denote the set of (Schur) stable matrices, and define the \emph{Lyapunov map}
\(\textstyle \lyap \colon \Astable \times \Amatrices \mapsto\Amatrices,\)
that sends the pair $(A,Q)$ to the unique solution $X$ of 
\begin{equation}\label{eq:lyap-gen}
X = A X A^\intercal + Q,
\end{equation}
which has the representation $\textstyle X = \sum_{i=0}^\infty A^i Q (A^\intercal)^i$;
in this case, if $Q \succeq 0\, (\succ 0)$, then $X \succeq 0 \, (\succ 0)$. Furthermore, when $Q \succeq 0$, then $X \succ 0$ if and only if $(A,Q^{1/2})$ is controllable (see \cite{gajic2008lyapunov} and references therein).
The following is a frequently used technical lemma.
\begin{lemma}[Lemma 3.1 in\cite{talebi2022policy}]\label[lemma]{lem:dlyap}
The subset $\Astable$ is an open submanifold of $\Amatrices$, the Lyapunov map $\lyap\colon \Astable \times \Amatrices \to \Amatrices$ is smooth, and its differential acts as 
\begin{gather*}
    \diff \lyap_{(A,Q)}[E,F] =
    \lyap \big(A, E \lyap(A,Q) A^\intercal + A \lyap(A,Q) E^\intercal + F \big)
\end{gather*}
on any $(E,F) \in T_{(A,Q)} (\Astable \times \Amatrices) \cong \Amatrices \oplus \Amatrices$.
Furthermore, for any $A \in \Astable$ and $Q, \Sigma \in \Amatrices$ we have, the so-called \emph{\lyaptrace} property,
\[\tr{\lyap(A^\intercal,Q) \Sigma} = \tr{\lyap(A, \Sigma) Q}.\]
\end{lemma}

\section{Geometric Regularization and The Algorithm}\label{sec:geometry-algorithm}
By the estimation-control duality established in \cite[Proposition 1]{talebi2022duality}, the mean-squared prediction error in \cref{eq:mse-x-L} takes the following form
 \begin{equation*}
    J^\text{est}_T(L)
    = \tr{X_T(L)H^\intercal H} + \tr{R}.
\end{equation*}
In order to streamline the analysis, we consider the steady state regime and thus define the set of Schur stabilizing gains
\[\mathcal{S} \coloneqq \{L \in \bR^{n\times m}: \rho(A-LH) < 1\}.\]
Now, for any $L \in \mathcal S$ in the steady-state limit as $T \to \infty$: 
\(X_T(L) \to X_{\infty}(L) \coloneqq \sum_{t=0}^{\infty} \left(A_L\right)^{t} \left(  Q + L R L^\intercal \right)\left(A_L^\intercal\right)^{t}.\)
Because $\rho(A_L) < 1$, the limit coincides with the unique solution $X_L = \lyap(A_L,Q + L R L^\intercal)$.
Therefore, the steady-state limit of the mean-squared prediction error $J_{\mathrm{MSE}}(L) \coloneqq \lim_{T \to \infty} J_{T}^\text{est}(L)$ is well-defined and in fact the convergence is exponentially fast in $T$.
Thus, we formally analyze the following constrained optimization problem:
\begin{align}\label{eq:opt-time-indepen}
    \min_{L \in \mathcal{S}} \; &\leftarrow J_{\mathrm{MSE}}(L) = \tr{\lyap(A_L,Q + L R L^\intercal) H^\intercal H} + \tr{R},
\end{align}

The pair $(A,H)$ is observable, so is $(A_L,H)$ for any $L\in\mathcal{S}$ (as any unobservable mode of $(A_L,H)$ is indeed an unobservable mode of $(A,H)$). Equivalently, $(A^\intercal,H^\intercal)$ is controllable and thus $\lyap(A_L^\intercal, H^\intercal H) \succ 0$ for all $L\in\mathcal{S}$.
Therefore, following \cite[Proposition 3.3]{talebi2022policy}, we equip $\mathcal{S}$ with a Riemannian metric. 
\[\langle V, W \rangle_{Y_L} \coloneqq \mathrm{tr}[V W^\intercal Y_L]\] 
where $Y_L = \lyap(A_L^\intercal, H^\intercal H)$. Here, we embed $\mathcal{S}$ into $\mathbb{R}^{n\times(n+m)}$ by sending $L\mapsto \begin{bmatrix} I & L \end{bmatrix}$ and equip this with the sub-Riemannian metric induced by the same $Y_L$. With abuse of notation, we use the same symbols $(\mathcal{S},\langle\cdot,\cdot\rangle_{Y_L})$ to denote this embedded manifold and its induced sub-Riemannian metric whenever convenient.

Next, by the Lyap-trace property, we can show that 
\begin{multline}
    J_{\mathrm{MSE}}(L) - \tr{R} = \tr{(Q + L R L^\intercal)Y_L} =\\ \tr{\begin{bmatrix}
        I & L
    \end{bmatrix}
    \begin{bmatrix}
        Q & 0 \\ 0 & R
    \end{bmatrix}
    \begin{bmatrix}
        I & L
    \end{bmatrix}^\intercal Y_L}  = \left\|\begin{bmatrix}
        I & L
    \end{bmatrix}
    \begin{bmatrix}
        Q^{\frac{1}{2}} & 0 \\ 0 & R^{\frac{1}{2}}
    \end{bmatrix}
    \right\|_{Y_L}^2 \label{eq:mse-as-2-norm}
\end{multline}
In particular with respect to our sub-Riemannian metric, we showed that the MSE cost can be viewed as a simple 2-norm of the filtering policy, rescaled with noise covariances.  

Inspired by this intuition, we use the same sub-Riemannian metric to introduce the Riemannian-Regularized MSE cost, denoted by $J_{\mathrm{R}}$, as
\begin{equation}\label{eq:regularized-problem}
    J_{\mathrm{R}}(L,\gamma) := J_{\mathrm{MSE}}(L) + \gamma \left\|\begin{bmatrix}
    I & L
\end{bmatrix}\right\|_{Y_L}^{2},
\end{equation}
with $\gamma>0$ being a regularization factor.
We will show how this Riemannian regularization recovers vital properties required for learning a policy, thus resulting in a well-conditioned learning problem. These properties will be justified lated in \Cref{prop:properties}.

\begin{lemma}\label{lem:regularized-gradient}
    The regularized cost $J_{\mathrm{R}}$ and its gradient takes the following explicit forms,
\begin{align}
    &J_{\mathrm{R}}(L,\gamma) = \tr{(Q_\gamma + L R_\gamma L^\intercal) Y_L} + \tr{R}\label{eq:regularized-cost-Y_L}\\
    &\nabla_L  J_{\mathrm{R}}(L,\gamma) 
    = 2Y_L \left(-LR_\gamma + A_L X_{(L,\gamma)} H^{\intercal} \right)\label{eq:regularized-grad},
\end{align}
where $Y_L = \lyap(A_L^\intercal, H^\intercal H)$ and $X_{(L,\gamma)} = \lyap(A_L, Q_\gamma + L^\intercal R_\gamma L)$, 
with $Q_\gamma = Q + \gamma I$ and $R_\gamma = R + \gamma I$.
\end{lemma}

\begin{proof}
    By the definition of $\langle\cdot,\cdot\rangle_{Y_L}$,
    
    \begin{align*}
    \gamma \left\|\begin{bmatrix}
        I & L
    \end{bmatrix}\right\|_{Y_L}^{2} = \tr{\gamma(I+LL^\intercal)Y_L}.
    \end{align*}
    Combining like terms with the first equality of  \cref{eq:mse-as-2-norm} recovers \cref{eq:regularized-cost-Y_L}. The gradient follows from the formula in \cite{talebi2022duality} by substituting $Q_\gamma$ and $R_\gamma$ in for $Q$ and $R$ respectively. 
\boxend\end{proof}

\subsection{The learning algorithm} \label{sec:the-algorithm}
For each fixed $\gamma>0$, we also characterize the global minimizer $L^* = \arg\min_{L \in \mathcal{S}} J_{\mathrm{R}}(L,\gamma)$. The domain $\mathcal{S}$ is non-empty whenever $(A,H)$ is observable. 
Thus, by continuity of $L \to J_{\mathrm{R}}(L,\gamma)$, there exists some finite $\alpha,\gamma > 0$ such that the regularized sublevel set $\mathcal{S}_{(\alpha,\gamma)}$ is non-empty and compact (see \Cref{prop:properties}). Therefore, the minimizer is an interior point and thus must satisfy the first-order optimality condition $\nabla J_{\mathrm{R}}(L_\gamma^*,\gamma) = 0$. Moreover, by coercivity of the regularized cost, the minimizer is stabilizing and unique, and satisfies
\(L_\gamma^* = A X_\gamma^* H^\intercal \left(R_\gamma + H X_\gamma^* H^\intercal\right)^{-1},\)
with $X_\gamma^* = \lyap(A_{L_\gamma^*}, Q_\gamma + L_\gamma^* R_\gamma (L_\gamma^*)^\intercal)$.
As expected, the regularized global minimizer $L_\gamma^*$ is explicitly dependent on the noise covariances $Q$ and $R$, and the regularizer $\gamma$. Based on this intuition, we provide the following algorithm as an extension of ideas in \cite{talebi2022duality,talebi2023data} via continuation for the regularized learning problem \cref{eq:opt-time-indepen} via \cref{eq:regularized-problem}:

\begin{algorithm}[ht]
\label{alg:algo_continuation}
\caption{Riemannian-Regularized Kalman Policy Optimization}
\begin{algorithmic}[1]  
\STATE Get a tolerance $\epsilon > 0$, failure probability $\delta$, $\gamma_{\min}>0$, and initialize regularization factor $\gamma_0 \gg \gamma_{\min}$  
\STATE Get a policy $L_0 \in \mathcal{S}$ with $\alpha \geq J_{\mathrm{R}}(L_0, \gamma_0)$
\STATE Set measurement length $T\geq \mathcal{O}\left(\ln(\epsilon^{-1})\right)$ and \\batch-size $M \geq \mathcal{O}\left(\epsilon^{-1}\ln(\ln(\epsilon^{-1})) \ln(\delta^{-1}) \right)$
\STATE Set $L \gets L_0$ and choose a stepsize $\eta \leq \frac{2}{9 \ell(\alpha,\gamma_{\min})}$
\FOR{$k=0,1,\cdots, K \geq \mathcal{O}(\ln(\epsilon^{-1}))$}
\STATE $\gamma \gets \gamma_k$ 
\WHILE{$\|\nabla_L J_{\mathrm{R}}(L, \gamma_k)\|^2 \geq \frac{(1-\beta) \alpha \gamma_k}{c(\alpha,\gamma_{\min})}$} 
\STATE Get $M$ independent measurements of length $T$:\\
\(\qquad \mathcal Y_{[0:T]}^i = \{y_t\}_{0}^{T}, \quad i=1,\cdots,M\)
\STATE From the gradient oracle, inquire: \\
$\qquad \nabla_L J_{\mathrm{R}}(L, \gamma) \gets \text{\textbf{Oracle}}(L,\{\mathcal Y_{\{0:T\}}\}_{i=1}^M)$
\STATE Update the policy:\\ 
$\qquad L \gets L - \eta \nabla_L J_{\mathrm{R}}(L, \gamma)$
\ENDWHILE 
\STATE Set the converged policy $L_{k+1} \gets L$
\STATE \textbf{Geometric schedule:} \\update $\gamma_{k+1} \gets \max\left[\gamma_{\min}, \beta \gamma_k \right]$ where $\beta \in (0, 1)$ 
\ENDFOR 
\RETURN $L_K$ 
\end{algorithmic} 
\end{algorithm}

The rationale behind the algorithm is that combining linear convergence within each continuation step with the geometric decay of $\gamma$ will result in a procedure that converges linearly to the unregularized solution $L^*$. The constants $\ell(\alpha,\gamma)$ and $c(\alpha,\gamma)$ correspond to the locally Lipschitz and PL properties of the regularized cost, respectively, as defined explicitly later in \Cref{prop:properties}.

The regularization idea in the context of estimation problems has been explored recently \cite{Khlebnikov2023comparison}, however, using a Euclidean $\ell_2$-regularization of the filtering gain and through a different approach using \textit{invariant ellipsoid}. Nonetheless, we compare our algorithm against such Euclidean $\ell_2$-regularization of the filtering policy in \S\ref{sec:simulation}.

\section{Data-driven Gradient Oracle}
\label{sec:data-driven-oracle}
When noise covariance matrices $Q$ and $R$ are unknown, it is not possible to directly compute the gradient of the MSE cost from \Cref{lem:regularized-gradient}.    
Therefore, we construct a stochastic gradient oracle that estimates the gradient from the data at hand. 
For that, consider a length $T-t_0$ sequence of measurements  $\mathcal Y_{\{t_0:T\}}:=\{y(t_0),y(t_0+1),\ldots, y(T-1)\}$ starting at some initial time $t_0$. 
Given any filtering gain $L \in \mathcal{S}$, using \cref{eq:estimate-x-L} we obtain an estimate $\hat{y}_L(T)$ as
\begin{equation*}\textstyle
  \hat{y}_L(T) = H \hat{x}_L(T)= H A_L^{T-t_0} m_0 +\sum_{t=t_0}^{T-1} H  A_L^{T-t-1} L y(t).
\end{equation*}
This results in an estimation error $e_T(L)\coloneq y(T) - \hat{y}_L(T)$ with its squared-norm denoted by 
\[\varepsilon(L,\mathcal Y_{\{t_0:T\}}) \coloneqq \|e_T(L)\|^2\]

Let us now consider the regularized mean squared-norm of the error over all possible random measurement sequences:  
\begin{equation*}
    J_{\{t_0:T\}}(L,\gamma):=\E{}{\varepsilon(L,\mathcal Y_{\{t_0:T\}})} + \gamma \left\|\begin{bmatrix}
    I & L
\end{bmatrix}\right\|_{Y_L}^{2}.
\end{equation*}
Then it is then straightforward to show that
\[\lim_{(T-t_0)\to \infty}J_{\{t_0:T\}}(L) = J_{\mathrm{R}}(L),\]
with an exponentially fast rate.

Next, assuming access to independent collection of the measurement sequence,  the gradient of the regularized MSE cost can be approximated as follows:
\begin{proposition}[Gradient Oracle]\label{prop:approx-grad}
    Given $L \in \mathcal{S}$ and $M$ independently collected measurements  $\{\mathcal Y_{[t_0,T]}^i\}_{i=1}^M$, define
    \begin{multline*}\textstyle
    \widehat{\nabla J_{\mathrm{R}}}(L,\gamma) = \frac{1}{M}\sum_{i=1}^M \nabla_L \varepsilon(L,\mathcal Y^i_{\{t_0:T\}}) \\+ 2 \gamma Y_L \left(-L + A_L Z_{L} H^{\intercal} \right),
\end{multline*}
where $Z_{L}=\lyap(A_L,I)$ and
\begin{multline*}
    \nabla_L \varepsilon(L,\mathcal Y)= -2\sum_{t=0}^{T-t_0-1} (A_L^\intercal)^{t} H^\intercal e_{T}(L) y^\intercal(T-t-1)    \\
    +2\sum_{t=1}^{T-t_0-1}\sum_{k=1}^{t} (A_L^\intercal)^{t-k} H^\intercal e_{T}(L) y^\intercal(T-t-1) L^\intercal (A_L^\intercal)^{k-1} H^\intercal .
\end{multline*}
Then,
$\widehat{\nabla J_{\mathrm{R}}}(L,\gamma)$ is an unbiased estimate of the gradient $\nabla J_{\mathrm{R}}(L,\gamma)$; i.e., as $T-t_0 \to \infty$,
\begin{equation*}\textstyle
     \E{}{\widehat{\nabla J_{\mathrm{R}}}(L,\gamma)} \to  \nabla J_{\mathrm{R}}(L,\gamma).
\end{equation*}
\end{proposition}
\begin{proof}
    Computing the regularizing norm in \cref{eq:regularized-problem} does not require knowledge of $Q$ or $R$. Thus, we focus on estimating $J_\mathrm{MSE}$.
    Denote the approximated MSE cost value by
    \begin{equation*}
    \widehat{J_\mathrm{MSE}}(L)\coloneq\frac{1}{M}\sum_{i=1}^M  \varepsilon(L,\mathcal Y^i_{\{t_0:T\}}).
    \end{equation*}

For small enough $\Delta \in \mathbb R^{n\times m}$,
\begin{multline*}
   \varepsilon(L+\Delta,\mathcal Y) - \varepsilon(L,\mathcal Y) =  \|e_{T}(L + \Delta) \|^2 - \|e_{T}(L) \|^2 =\\
    2 \tr{(e_{T}(L + \Delta) - e_{T}(L) )e_{T}^\intercal(L)} + o(\|\Delta\|)).
\end{multline*}
The difference 
\begin{gather*}
    e_{T}(L+\Delta) - e_{T}(L) = E_1(\Delta) + E_2(\Delta) + o(\|\Delta\|),
\end{gather*}
contains the following terms that are linear in $\Delta$:
\begin{align*}
    E_1(\Delta) &\coloneqq \textstyle -\sum\limits_{t=0}^{T-t_0-1} H  (A_L)^{t} \Delta y(T-t-1),\\
    E_2(\Delta) &\coloneqq \textstyle \sum\limits_{t=1}^{T-t_0-1}\sum
    \limits_{k=1}^{t} H  (A_L)^{t-k} \Delta H (A_L)^{k-1} L y(T-t-1).
\end{align*}
Therefore, combining the two identities, the definition of gradient under the inner product $\langle A,B\rangle :=\tr{AB^\intercal}$, and ignoring the higher order terms in $\Delta$ yields,
\begin{gather*}
    \langle \nabla_L \varepsilon(L,y), \Delta \rangle =  2 \tr{(E_1(\Delta)+ E_2(\Delta) )e_{T}^\intercal(L)},
\end{gather*}
which by linearity and cyclic permutation property of trace reduces to:
\begin{gather*}
    \langle \nabla_L \varepsilon(L,y), \Delta \rangle = - 2 \tr{\Delta \left(\sum_{t=0}^{T-t_0-1} y(T-t-1) e_{T}^\intercal(L) H (A_L)^{t} \right)}\\
    + 2 \tr{\Delta \left(\sum_{t=1}^{T-t_0-1}\sum_{k=1}^{t} H (A_L)^{k-1} L y(T-t-1) e_{T}^\intercal(L) H (A_L)^{t-k} \right)}.
\end{gather*}
This holds for all admissible $\Delta$, concluding the formula for the gradient.
\boxend\end{proof}

Note that the variance of this estimate also converges to zero at a rate of $O(\frac{1}{M})$ as the number of sample measurements $M$ increases. The number $M$ is referred to as the batch-size. 

Noting that this gradient estimate $\widehat{\nabla J_{\mathrm{R}}}(L,\gamma)$ only depends on the available information, we utilize this gradient approximation as the gradient oracle in Algorithm 1. The convergence under this stochastic oracle can be obtained using the  gradient dominance condition and locally Lipschitz property, but the analysis becomes more complicated than the deterministic oracle due to the possibility of the iterated gain $L_k$ leaving the sub-level sets due to stochastic errors in the gradient.

\section{Convergence Analysis} \label{sec:convergence}
First, we establish how this geometric regularization enables us to recover the essential properties required for the learning the optimal Kalman gain through a direct policy optimization.

\begin{proposition}\label{prop:properties}
    Suppose $(A,H)$ is observable, and consider the Riemannian regularized MSE cost $J_{\mathrm{R}}:\mathcal{S}\times \mathbb{R}_{\geq0} \to \mathbb{R}$. Then, the following holds true for each fixed $\gamma \geq 0$:
    \begin{enumerate}
        \item The Riemannian regularized cost $J_{\mathrm{R}}(L,\gamma)$ is coercive in $L$ if $\gamma>0$; i.e., for any sequence $\{L_k\} \in \mathcal{S}$,
        \[\text{ if~~} L_k \to \partial\mathcal{S} \text{~or~} \|L_k\| \to \infty \text{~then~} J_{\mathrm{R}}(L_k,\gamma) \to \infty.\]
        
        \item For any $\alpha >0$, the sublevel set $\mathcal{S}_{(\alpha,\gamma)} \coloneqq \{L \in \bR^{n\times m}: J_{\mathrm{R}}(L,\gamma) \leq \alpha\}$ is compact if $\gamma>0$, is contained in $\mathcal{S}$, and $\mathcal{S}_{(\alpha,\gamma_1)} \subset \mathcal{S}_{(\alpha,\gamma_2)}$ whenever $0 \leq \gamma_2 \leq \gamma_1$.

        \item There exists a unique global minimizer of $J_{\mathrm{R}}(\cdot,\gamma)$ denoted by \[L_\gamma^* = A X_\gamma^* H^\intercal \left(R_\gamma + H X_\gamma^* H^\intercal\right)^{-1},\]
        with $X_\gamma^* = \lyap(A_{L_\gamma^*}, Q_\gamma + L_\gamma^* R_\gamma (L_\gamma^*)^\intercal)$. 
        
        \item The Riemannian regularized cost $J_{\mathrm{R}}(\cdot,\gamma)$ has the PL-property: $\forall L \in \mathcal{S}_{(\alpha,\gamma)}$
        \[J_{\mathrm{R}}(L,\gamma) - J_{\mathrm{R}}(L_\gamma^*,\gamma) \leq c(\alpha,\gamma) \|\nabla_L J_{\mathrm{R}}(L,\gamma)\|^2\]
        where 
        \[c(\alpha,\gamma) \coloneqq \frac{\lambdamax(Y_L^*)}{2\gamma\kappa_{(\alpha,\gamma)}^2}\]
        and $\kappa_{(\alpha,\gamma)}\coloneq\inf_{L\in \mathcal{S}_{(\alpha,\gamma)}}\lambdamin(Y_L)>0.$ Also, $c(\alpha,\gamma)$ is decreasing in $\gamma$ for any fixed $\alpha$.
        
        \item The Riemannian regularized cost $J_{\mathrm{R}}$ has Lipschitz gradient on sublevel sets: $\forall L,L' \in \mathcal{S}_{(\alpha,\gamma)}$
        \[\|\nabla_L J_{\mathrm{R}}(L,\gamma) - \nabla_L J_{\mathrm{R}}(L',\gamma)\| \leq \ell(\alpha,\gamma) \|L - L'\|,\]%
        where
        \[
            \ell(\alpha,\gamma)\coloneq\frac{2\alpha}{\gamma}(\lambdamax(R)+\gamma+\alpha)+\frac{4\sqrt{2}\alpha}{\gamma}\|H\|_2\xi_{(\alpha,\gamma)},
        \]
        is decreasing in $\gamma$ for any fixed $\alpha$, and $\xi_{(\alpha,\gamma)}$ is a non-increasing function of $\gamma$ defined in the proof.
    \end{enumerate}
\end{proposition}
\begin{proof}
    Lemmas 1 and 2 of \cite{talebi2023data} establish coercivity and gradient dominance respectively for the case of $Q,R\succ0$. Following the same argument using the form of the regularized cost from \cref{eq:regularized-cost-Y_L} in \Cref{lem:regularized-gradient} and noting that $Q_\gamma,R_\gamma\succ0$ shows parts 1 through 4.
    For part 5, like in the proof of the dual LQR problem's gradient dominance \cite[Proposition 3.10]{bu2019lqr}, we can show that the Hessian of $J_{\mathrm{R}}$ is characterized by
    \begin{align*}
        \diff^2J_{\mathrm{R}}(L,\gamma)[E,E]=&2\tr{(R_\gamma E^\intercal+HX_{(L,\gamma)}H^\intercal E^\intercal Y_L)E}\\
        &-4\tr{H\diff X_{(L,\gamma)}[E]A_L^\intercal Y_LE}\\
        \eqqcolon&2\tr{a_{(L,\gamma)}}-4\tr{b_{(L,\gamma)}}.
    \end{align*}
    We want to bound the magnitude of the Hessian, so consider
    \begin{align*}
        &\max\limits_{L\in{\mathcal{S}_{(\alpha,\gamma)}}}\sup\limits_{\|E\|_F=1}|\diff^2J_{\mathrm{R}}(L)[E,E]|\\
        &\leq\max\limits_{L\in{\mathcal{S}_{(\alpha,\gamma)}}}\sup\limits_{\|E\|_F=1}(2|\tr{a_{(L,\gamma)}}|)+\sup\limits_{\|E\|_F=1}(4|\tr{b_{(L,\gamma)}}|).
    \end{align*}
    Let $E$ by any unit norm tangent vector. Because $\tr{Y_L}\leq\alpha/\gamma$ and $\tr{X_{(L,\gamma)}H^\intercal H}=J_{\mathrm{R}}(L,\gamma)\leq\alpha$,
    \begin{align}   |\tr{a_{(L,\gamma)}}|\leq\alpha(\lambdamax(R)+\gamma+\alpha)/\gamma.\label{eq:a_upperbound} 
    \end{align}
    Clearly, $a_{(\alpha,\gamma)}$ is decreasing in $\gamma$. We now consider the second term. From \cite[Proposition 7.7]{bu2019lqr}, 
    \begin{align}
        \tr{Y_L^{1/2}}\leq\sqrt{(2\alpha/\gamma)},\text{  and  }\|A_LY_L^{1/2}\|_2\leq\sqrt{\alpha/\gamma}.\label{eq:prop2.1LQR}
    \end{align}
    
    To deal with the $\diff X_{(L,\gamma)}[E]$ term, we define the following terms that are non-increasing in $\gamma$:
    \[
        \mathcal{X}_{(\alpha,\gamma)}\coloneq\max\limits_{L\in \mathcal{S}_{(\alpha,\gamma)}}\tr{X_L} \text{, and }
        \mathcal{Z}_{(\alpha,\gamma)}\coloneq\max\limits_{L\in \mathcal{S}_{(\alpha,\gamma)}}\tr{\lyap(A_L,I)}.  
    \]
    Recall that $\diff X_{(L,\gamma)}[E]=\lyap(A_L,V)$ for a matrix $V$ depending on the problem parameters. To bound $V$, consider that
    \begin{align}
        &\diff X_{(L,\gamma)}[E]-A_L\diff X_{(L,\gamma)}[E]A_L^\intercal \notag\\
        &=-EHX_{(L,\gamma)}A_L^\intercal-A_LX_{(L,\gamma)}(EH)^\intercal+ER_\gamma L^\intercal+LR_\gamma E^\intercal \notag\\
        &\preceq A_LX_{(L,\gamma)}A_L^\intercal+(EH)X_{(L,\gamma)}(EH)^\intercal+ER_\gamma E^\intercal+LR_\gamma L^\intercal \notag\\
        &=X_{(L,\gamma)}-Q_\gamma +(EH)X_{(L,\gamma)}(EH)^\intercal+ER_\gamma E^\intercal \notag\\
        &\preceq(\tr{X_{(L,\gamma)}}+\alpha+\lambdamax(R)+\gamma)Q_\gamma/\gamma,\label{eq:dX_second_arg_upperbound}
    \end{align}
    where in the first inequality we used part (b.1) of \cite[Proposition 2.1]{bu2019lqr}. By part (c) of the same proposition, this implies that $\diff X_{(L,\gamma)}[E]$ is less than or equal to the following in Loewner ordering:
    \begin{align}
        \frac{1}{\gamma}(\tr{X_{(L,\gamma)}}+\alpha+\lambdamax(R)+\gamma)(X_L+\gamma\mathcal{Z}_{(\alpha,\gamma)}I).\label{eq:dX_Loewner_bound}
    \end{align}
    Using (b.2) of \cite[Proposition 2.1]{bu2019lqr} in the first inequality of \cref{eq:dX_second_arg_upperbound} instead of (b.1) and combining the result with \cref{eq:dX_Loewner_bound} shows that 
    \begin{align}
        &\|\diff X_{(L,\gamma)}[E]\|_2 \label{eq:dX_operator_norm_bound}\\
        &\leq\frac{1}{\gamma}[\tr{X_{(L,\gamma)}}+\alpha+\lambdamax(R)+\gamma](\mathcal{X}_{(\alpha,\gamma)}+\gamma\mathcal{Z}_{(\alpha,\gamma)}). \notag
    \end{align}
    To address the $\tr{X_{(L,\gamma)}}$ term, it will be helpful to observe that $\|L\|_F^2\leq\alpha/(\gamma\kappa_{(\alpha,\gamma)})$, which follows from the fact that $L\in \mathcal{S}_{(\alpha,\gamma)}$. Thus 
    $\tr{\lyap(A_L,LL^\intercal)}\leq\tr{\lyap(A_L,\tr{LL^\intercal}I)}\leq[\alpha/(\gamma\kappa_{(\alpha,\gamma)})]\mathcal{Z}_{(\alpha,\gamma)}.$ From this inequality and the linearity of $\lyap$ in its second argument,
    \begin{align}
        \tr{X_{(L,\gamma)}}&=\tr{X_L}+\gamma\tr{\lyap(A_L,I+LL^\intercal)} \notag\\
        &\leq\mathcal{X}_{(\alpha,\gamma)}+(\gamma+\alpha/\kappa_{(\alpha,\gamma)})\mathcal{Z}_{(\alpha,\gamma)}.
        \label{eq:tr_X_gamma_upperbound}
    \end{align}
    Substituting \cref{eq:tr_X_gamma_upperbound} into \cref{eq:dX_second_arg_upperbound}, we can see that $\|\diff X_{(L,\gamma)}[E]\|_2$ is upper bounded by a function that is decreasing in $\gamma$. Let $\xi_{(\alpha,\gamma)}$ denote this function. Then by \cref{eq:prop2.1LQR}, \cref{eq:dX_operator_norm_bound}, and \cref{eq:tr_X_gamma_upperbound},
    \begin{align}
        |\tr{b_{(L,\gamma)}}|&\leq\|H\|_2\|\diff X_{(L,\gamma)}[E]\|_2\|A_LY_L^{1/2}\|_2\tr{Y_L^{1/2}} \notag \\
        &\leq\sqrt{2}\|H\|_2\alpha\xi_{(\alpha,\gamma)}/\gamma, \label{eq:b_upperbound}
    \end{align}
    which is a decreasing function of $\gamma$. Combining \cref{eq:a_upperbound} and \cref{eq:b_upperbound} justifies the claimed Lipschitz constant. 
\boxend\end{proof}

Finally, these results are sufficient to provide recursive feasibility and convergence guarantees for Algorithm 1 which recovers the globally optimal filtering policy from measurement data despite the lack of fully excited noise. To simplify our probabilistic analysis, we consider almost surely bounded measurement and process noise with zero mean; the extension to sub-Gaussian noise follows by a standard argument involving Bernstein's concentration inequalities. 

\begin{theorem}
Suppose $(A,H)$ is observable. Assume that $\|x_0\|$, $\|\xi(t)\|\leq \kappa_\xi$, and $\|\omega(t)\|\leq \kappa_\omega$ for all $t$ (almost surely), and the initial state has zero mean, i.e., $m_0 = 0$. Consider $L_0\in\mathcal{S}$, fix $\gamma_0\geq \gamma_{\min}>0$ and $\alpha \geq J_{\mathrm{R}}(L_0,\gamma_0)$, and set stepsize 
\(\eta=\frac{2}{9 \ell(\alpha,\gamma_{\min})}.\)
Then, for any $\gamma\geq\gamma_{\min}$ and $\epsilon,\delta>0$, with probability at least $1-\delta$ the internal loop terminates in $\mathcal{O}(\ln(1/\epsilon))$ iterations if 
\[T\geq \mathcal{O}\left(\ln(\frac{1}{\epsilon})\right)\quad \text{and}\quad M \geq \mathcal{O}\left(\ln(\frac{1}{\delta})\frac{1}{\epsilon}\ln(\ln(\frac{1}{\epsilon})) \right).\]
Furthermore, the optimality gap decays linearly as 
\[J_{\mathrm{R}}(L_k,\gamma_k) - J_{\mathrm{MSE}}(L^*) \leq \alpha (2-\beta)\max\{\gamma_0 \beta^k,\gamma_{\min}\}.\]\label{thm:convergence}
\end{theorem}
\begin{proof}
    Under the hypothesis and the choice of stepsize, we can show that that Assumption 2, 3 and 4 in \cite{talebi2023data} are all satisfied if $\gamma \geq \gamma_{\min}>0$. Therefore, if the trajectory length $T$ and batch-size $M$ is large enough (in particular satisfying the rates stated above), by \cite[Theorem 3]{talebi2023data} the inner-loop converges in a linear rate. We compute this rate explicitly and then show that the outer-loop also has a linear convergence rate---due to geometric scheduling of the regularizer.
    
    Set \(f_k(L)=J_{\mathrm{R}}(L,\gamma_k)\), and denote by \(L_k^\star\) its minimizer. Let $L_{k,m}$ be the policy obtained by Algorithm 1 after $m$ inner iteration and $k$ outer iterations, and let $L_k$ be the output of the $k$-th inner loop.
    Recall the PL constant $c(\alpha,\gamma)$ defined in \Cref{prop:properties} and note that $c(\alpha,\cdot)$ is a decreasing function.
    Therefore, by \cite[Theorem 1]{talebi2023data}, we have
    \begin{equation*}
        f_k(L_{k,m}) - f_k(L_{k}^*) 
        \leq  q^m \left(f_k(L_{k,0}) - f_k(L_{k}^*) \right)
    \end{equation*}
    and terminated by
     \[
    \|\nabla f_k(L_k)\|_F^2 \le \frac{\varepsilon_k}{c(\alpha,\gamma_{\min})},
    \qquad
    \varepsilon_k \coloneqq C\,\gamma_k
    \]
    for some constant \(C>0\) independent of \(k\), with 
    \[q \coloneqq 1-\frac{\eta}{4 c(\alpha,\gamma_{\min})}\in(0,1).\]
    Define the gap $E_k \coloneq f_k(L_k) - f_k(L_k^*)$.
    At the outer stage \(k+1\), the inner loop is run for \(m\) gradient steps such that
    \begin{align}\label{eq:final-error-bound}
        E_{k+1} 
        \leq q^m \left(f_{k+1}(L_{k}) - f_{k+1}(L_{k+1}^*) \right)
    \end{align}
    because $L_{k+1,0} = L_k$.
    By \Cref{lem:regularized-gradient} and the definition of $L_{k+1}^*$ we have that
    \(f_{k+1}(L_{k+1}^*) \leq f_{k+1}(L_0) \leq f_{0}(L_0) \leq \alpha \)
    whenever $\gamma_{0} \geq \gamma_{k}\geq \gamma_{k+1}\geq \gamma_{\min}$, because $J_{\mathrm{R}}(L_0,\gamma) \leq f_0(L_0) \leq \alpha$ for all $\gamma\leq \gamma_0$; and thus,
    \begin{align}
    f_k(L_{k+1}^*)& -f_{k+1}(L_{k+1}^*) \nonumber\\
    &= (\gamma_k-\gamma_{k+1})\tr{[I + (L_{k+1}^*)^\intercal L_{k+1}^*] Y_{(L_{k+1}^*)}} \nonumber\\
    &\leq (1-\beta) \gamma_k f_{k+1}(L_{k+1}^*) \leq (1-\beta) \gamma_k \alpha \label{eq:delta-gamma-upperbound}
    \end{align}
    Also, by definitions of $f_k$, $L_k$, $L_k^*$, and $L_{k+1}^*$ we have
     \(
      f_{k+1}(L_k) -  f_k(L_k) \leq 0\) and \(
    f_k(L_k^*) - f_k(L_{k+1}^*) \leq 0.
    \)
    Therefore, by aggregating the last three upperbounds we obtain 
    \begin{align}\label{eq:initial-error-bound}
        f_{k+1}(L_{k})& - f_{k+1}(L_{k+1}^*)
        \leq f_k(L_k) - f_k(L_k^*) + (1-\beta)\alpha \gamma_k.
    \end{align}
    By combining \cref{eq:final-error-bound,eq:initial-error-bound} we obtain that
    \[E_{k+1} \leq q^m (E_k + (1-\beta)\alpha \gamma_k). \]
    Now, if $E_k \leq C_0 \gamma_k$ for some constant $C_0$ satisfying
    \[C_0 \geq (1-\beta)\alpha\]
    and $q^m \leq \beta/2$ then
    \[E_{k+1} \leq \frac{\beta}{2}(C_0 + (1-\beta)\alpha)\gamma_k \leq C_0 \gamma_{k+1}, \]
    as $\gamma_{k+1}=\beta \gamma_k$.
    Note that by PL property and the termination condition of the inner loop 
    \[E_k \leq c(\alpha,\gamma_{\min}) \|\nabla f_k(L_k)\|_F^2 \leq C_0 \gamma_k,\]
    and thus, by induction, we obtain the convergence rate  
    \[E_k \leq C_0 \gamma_k, \quad \forall k =1,2,\cdots.\]
    Finally,
    \begin{align*}
    J_{\mathrm{R}}(L_k,\gamma_k) - & J_{\mathrm{MSE}}(L^*)\\
        &=E_k + f_k(L_k^*) - f_k(L^*) + f_k(L^*) - J_{\mathrm{MSE}}(L^*)\\
        &\leq E_k + \gamma_k \alpha \leq (C_0 + \alpha)\gamma_k
    \end{align*}
    because $f_k(L_k^*) - f_k(L^*)\leq 0$ and $f_k(L^*) - J_{\mathrm{MSE}}(L^*) \leq \gamma_k \alpha$ by a similar argument as in \cref{eq:delta-gamma-upperbound}.
    Because \(\gamma_k=\max\{\gamma_0\beta^k,\gamma_{\min}\}\) the final claim follows, and for small enough $\gamma_{\min}>0$ the condition \(f_k(L_k) - f_0(L^*)\le \epsilon\) is guaranteed once
    \[\qquad\qquad K\geq \frac{\log(\epsilon/((C_0+\alpha)\gamma_0))}{\log \beta} = \mathcal{O}(\ln(\frac{1}{\epsilon})). \qquad\qquad \boxend\]
\end{proof}

\section{Simulations}\label{sec:simulation}
Here, we demonstrate the effectiveness of the proposed framework in improving estimation policies for a linear time-invariant (LTI) system.
Specifically, we consider a system with known dynamics $(A,H)$, where $n=4$ and $m=3$. Additional details on the generation of these system matrices are provided in the accompanying GitHub repository~\cite{talebi2026riemannian}. To deliberately construct an ill-conditioned estimation problem, the noise covariance matrices $Q$ and $R$, as well as the matrix $H^\intercal H$, are chosen to be singular.

Fixing the trajectory length $T=50$, we run Algorithm~1 using the stochastic gradient oracle in \Cref{prop:approx-grad}, with data from \cref{eqn:sysdyn}, and evaluate performance across varying batch sizes $M$. Then, fixing $M=20$, we also vary $T$ and run Algorithm 1.
As the choice of the trajectory length $T$ and the batch size $M$ affects the convergence behavior of Algorithm 1, we report the progress of the normalized MSE cost in \Cref{fig:sgd}, where each figure shows statistics over 50 rounds of simulation. Each round of simulation contains 20 continuation steps of 2 thousand iterations each, where $\gamma$ is scaled geometrically by a factor of $\beta=0.5$ each step. Furthermore, the normalized MSE cost is calculated as the performance of the current gain $L_K$ on the \textit{unregularized} objective $J_\mathrm{MSE}$ relative to the optimal solution $L^*$ to $J_\mathrm{MSE}$.

The results exhibit an initial phase of linear convergence, consistent with the theoretical guarantees established for the regularized objective under sufficiently accurate gradient estimates. As the iterates approach a neighborhood of the optimal solution, the convergence rate transitions to sublinear behavior. This degradation is expected, as the updates rely on stochastic approximations of the gradient, and the effect of estimation noise becomes dominant near optimality—contrasting with the linear convergence observed for \ac{gd} under an exact gradient oracle.
\Cref{prop:properties} illustrates convergence of the learned Kalman gain toward the optimal solution, in agreement with the structural properties of the objective function analyzed in \S\ref{sec:convergence}, particularly the gradient dominance condition established in \Cref{prop:properties}.

Finally, \Cref{fig:geometric_vs_euclidean} compares the performance of a conventional Euclidean $\ell_2$-regularization with the proposed Riemannian regularization on a structured problem which highlights the latter’s robustness, especially when as $\|L^*\|_F$ becomes larger:
\begin{gather*}
A = \begin{bmatrix}
    z & 1 & 0 \\
    0 & 0.5 & 1 \\
    0 & 0 & 0.5
\end{bmatrix},
\;H = \begin{bmatrix}
    1 & 0 & 0 \\
    0 & 1 & 0
\end{bmatrix},
\;Q = \begin{bmatrix}
    1 & 0 & 0 \\
    0 & 1 & 0 \\
    0 & 0 & 0
\end{bmatrix},
\;R = \begin{bmatrix}
    1 & 0 \\
    0 & 0
\end{bmatrix}
,
\end{gather*}
where $Q,R$, and $H^\intercal H$ are singular, $\gamma=0.1$, $\beta=0.25$, and $z$ is a hyperparameter controlling (proportionally) $\|L^*\|_F$. For each trial, we initialize a stabilizing gain $L_0$ with $z-0.5$ in the upper-left entry and zeros elsewhere, then run Algorithm 1 with the deterministic gradient oracle---to isolate the effect of regularization---for $20$ steps with 1000 inner iterations. For each $z$, the stepsize $\eta$ is the largest power of $10$ that results in convergence for both regularization types. The results demonstrate that for problems where the optimal gain is far from the origin, the Euclidean regularization fails to quickly converge to the optimal solution, as the indiscriminate penalty on $\|L\|_F$ drives the regularized solution away from $L^*$ and towards $0$. In contrast, Riemannian regularization converges more directly towards the optimal gain, even for large values of $z$ which result in an $L^*$ far from the origin, which aligns with the linear convergence of Algorithm 1 shown in \Cref{thm:convergence}. This highlights the benefit of the proposed Riemannian regularization and how compatible it is with the inherent geometry of the problem.

\begin{figure*}[t]
\vspace{-0.3cm}
    \centering
    \begin{subfigure}[b]{0.32\textwidth}
        \centering
        \includegraphics[width=\linewidth]{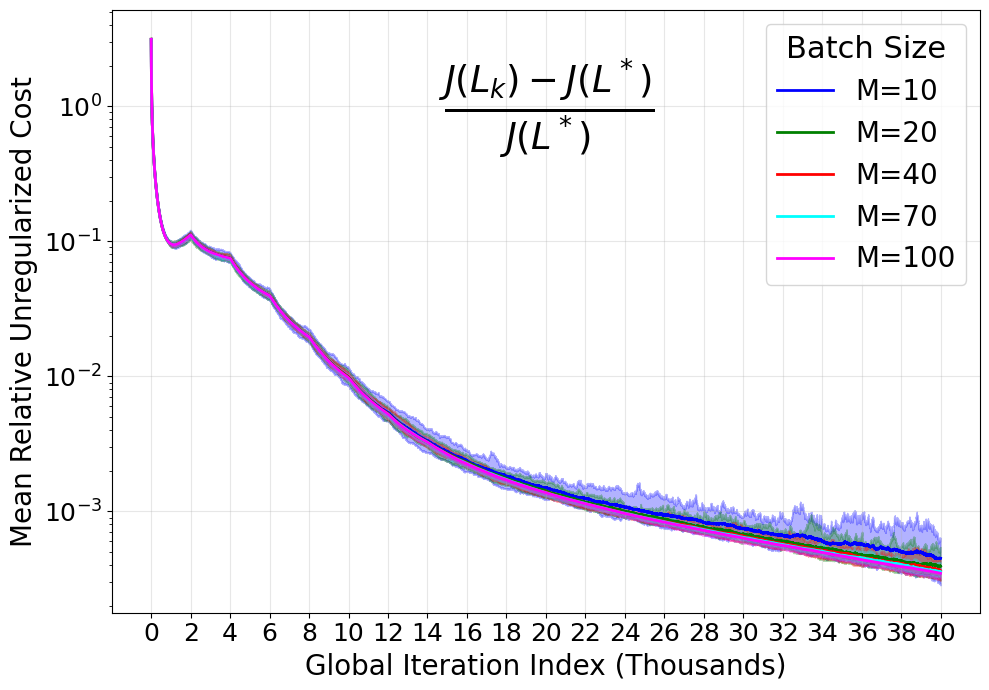}
        \caption{}
        \label{fig:a}
    \end{subfigure}
    \hfill
    \begin{subfigure}[b]{0.32\textwidth}
        \centering
        \includegraphics[width=\linewidth]{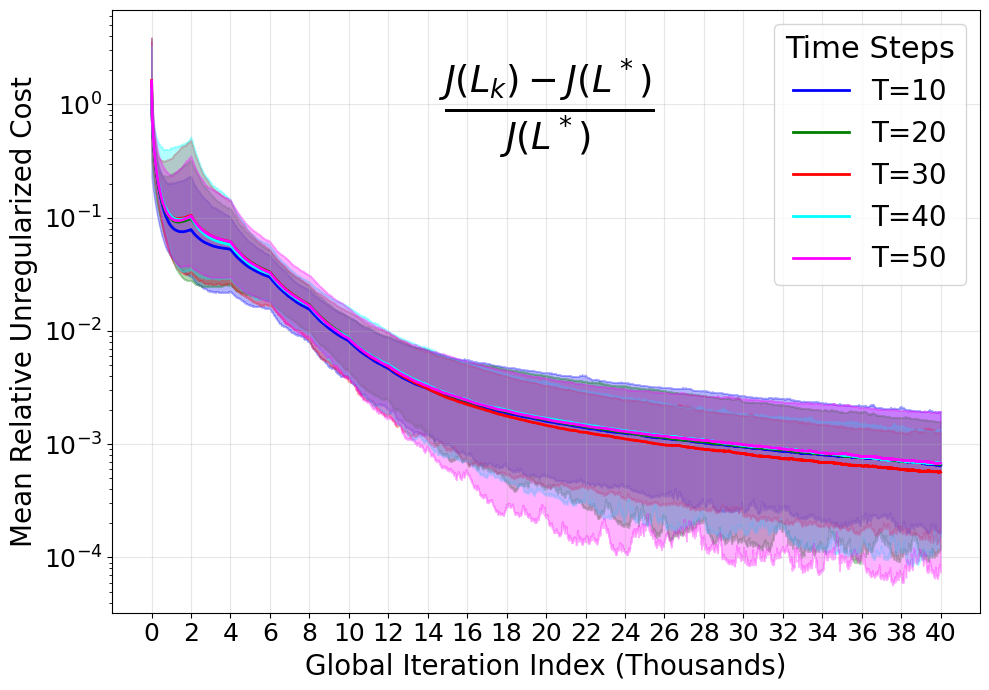}
        \caption{}
        \label{fig:b}
    \end{subfigure}
    \hfill
    \begin{subfigure}[b]{0.32\textwidth}
        \centering
        \includegraphics[width=\linewidth]{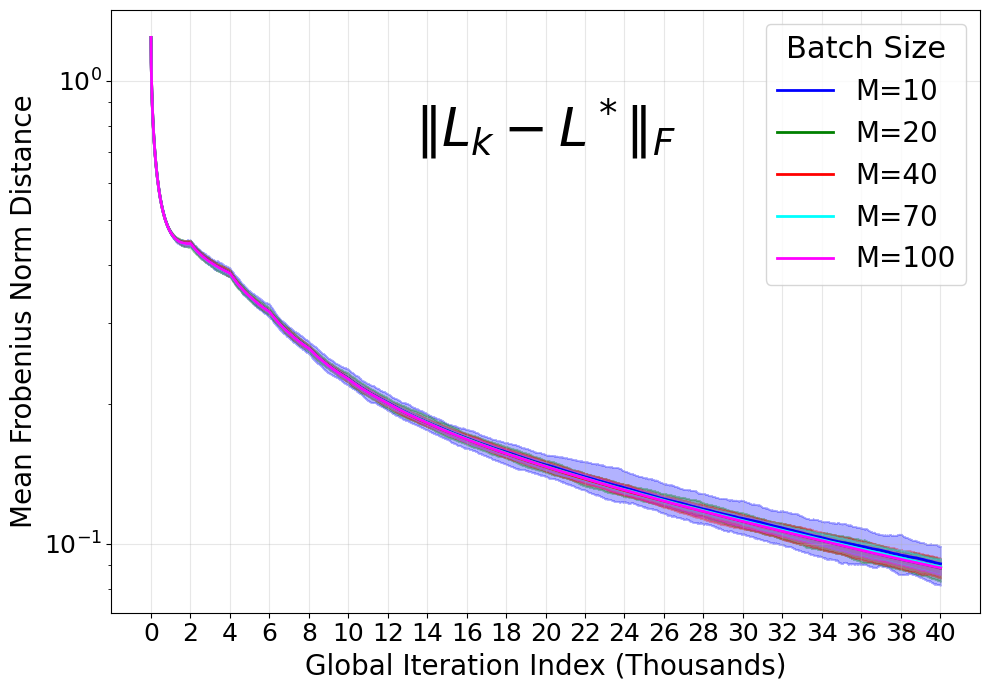}
        \caption{}
        \label{fig:c}
    \end{subfigure}
    \caption{\small Performance of our Riemanian-Regularized Kalman Policy Optimization (Algorithm 1) with the data-driven oracle and without knowledge of the \emph{singular} covariances. Vertical lines represent a new continuation step. Plots are of the mean progress over 50 trials with random initializations of (a) the estimation error for different batch sizes $M$, (b) the estimation error for different trajectory length $T$, and (c) the mean Frobenius norm distance between the $L_k$ and the optimal, un-regularized gain $L^*$ for various batch sizes. Shaded regions show 99\%-cover of the realizations per iteration step.}
    \label{fig:sgd}
    \vspace{-0.6cm}
\end{figure*}

\begin{figure}[t]
    \centering
    \includegraphics[width=0.35\textwidth]{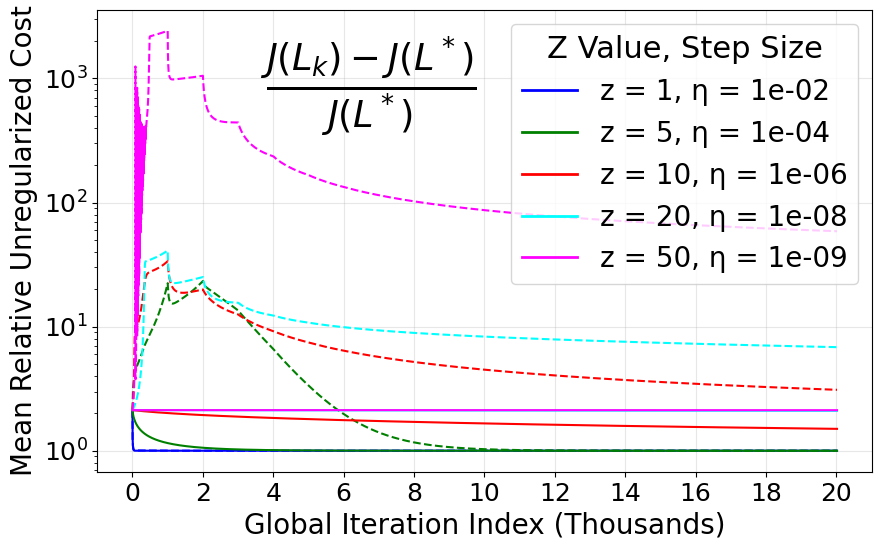}
    \caption{\small The benefit of the Riemannian regularization over the Euclidean one. Comparison of the convergence of the normalized, unregularized cost for various values of $z$ using the deterministic gradient oracle. Dashed lines illustrate the Euclidean $\ell_2$-regularization vs solid lines representing the Riemannian regularization.}
    \label{fig:geometric_vs_euclidean}
    \vspace{-0.6cm}
\end{figure}

\section{Conclusions}\label{sec:conclusion}
We studied the problem of learning the optimal steady-state Kalman gain in settings where the process and measurement noise covariances are both unknown and potentially singular, leading to an ill-conditioned problem. Leveraging the intrinsic geometry of the policy space, we introduced a Riemannian regularization that restores key structural properties of the objective, and developed a direct policy optimization algorithm based on a continuation scheme.
Our theoretical analysis establishes convergence guarantees for the proposed method, while empirical results demonstrate improved stability and performance compared to conventional Euclidean $\ell_2$-regularization. These findings highlight the effectiveness of incorporating geometric structure into data-driven estimation.
Future work will focus on extending this framework to account for model uncertainty, time-varying dynamics, and more general stochastic settings.
\bibliographystyle{ieeetr}
\bibliography{citations}

\begin{thebibliography}{10}

\bibitem{kalman1960new}
R.~E. Kalman, ``A new approach to linear filtering and prediction problems,''
  {\em ASME. Journal of Basic Engineering}, vol.~82, pp.~35--45, 03 1960.

\bibitem{mehra1970identification}
R.~Mehra, ``On the identification of variances and adaptive {Kalman}
  filtering,'' {\em IEEE Transactions on Automatic Control}, vol.~15, no.~2,
  pp.~175--184, 1970.

\bibitem{mehra1972approaches}
R.~Mehra, ``Approaches to adaptive filtering,'' {\em IEEE Transactions on
  Automatic Control}, vol.~17, no.~5, pp.~693--698, 1972.

\bibitem{carew1973identification}
B.~Carew and P.~Belanger, ``Identification of optimum filter steady-state gain
  for systems with unknown noise covariances,'' {\em IEEE Transactions on
  Automatic Control}, vol.~18, no.~6, pp.~582--587, 1973.

\bibitem{belanger1974estimation}
P.~R. Belanger, ``Estimation of noise covariance matrices for a linear
  time-varying stochastic process,'' {\em Automatica}, vol.~10, no.~3,
  pp.~267--275, 1974.

\bibitem{myers1976adaptive}
K.~Myers and B.~Tapley, ``Adaptive sequential estimation with unknown noise
  statistics,'' {\em IEEE Transactions on Automatic Control}, vol.~21, no.~4,
  pp.~520--523, 1976.

\bibitem{tajima1978estimation}
K.~Tajima, ``Estimation of steady-state {Kalman} filter gain,'' {\em IEEE
  Transactions on Automatic Control}, vol.~23, no.~5, pp.~944--945, 1978.

\bibitem{zhang2020identification}
L.~Zhang, D.~Sidoti, A.~Bienkowski, K.~R. Pattipati, Y.~Bar-Shalom, and D.~L.
  Kleinman, ``On the identification of noise covariances and adaptive {Kalman}
  filtering: A new look at a 50 year-old problem,'' {\em IEEE Access}, vol.~8,
  pp.~59362--59388, 2020.

\bibitem{magill1965optimal}
D.~Magill, ``Optimal adaptive estimation of sampled stochastic processes,''
  {\em IEEE Transactions on Automatic Control}, vol.~10, no.~4, pp.~434--439,
  1965.

\bibitem{hilborn1969optimal}
C.~G. Hilborn and D.~G. Lainiotis, ``Optimal estimation in the presence of
  unknown parameters,'' {\em IEEE Transactions on Systems Science and
  Cybernetics}, vol.~5, no.~1, pp.~38--43, 1969.

\bibitem{matisko2010noise}
P.~Matisko and V.~Havlena, ``Noise covariances estimation for {Kalman} filter
  tuning,'' {\em IFAC Proceedings Volumes}, vol.~43, no.~10, pp.~31--36, 2010.

\bibitem{kashyap1970maximum}
R.~Kashyap, ``Maximum likelihood identification of stochastic linear systems,''
  {\em IEEE Transactions on Automatic Control}, vol.~15, no.~1, pp.~25--34,
  1970.

\bibitem{shumway1982approach}
R.~H. Shumway and D.~S. Stoffer, ``An approach to time series smoothing and
  forecasting using the {EM} algorithm,'' {\em Journal of Time Series
  Analysis}, vol.~3, no.~4, pp.~253--264, 1982.

\bibitem{odelson2006new}
B.~J. Odelson, M.~R. Rajamani, and J.~B. Rawlings, ``A new autocovariance
  least-squares method for estimating noise covariances,'' {\em Automatica},
  vol.~42, no.~2, pp.~303--308, 2006.

\bibitem{aakesson2008generalized}
B.~M. {\AA}kesson, J.~B. J{\o}rgensen, N.~K. Poulsen, and S.~B. J{\o}rgensen,
  ``A generalized autocovariance least-squares method for {K}alman filter
  tuning,'' {\em Journal of Process Control}, vol.~18, no.~7-8, pp.~769--779,
  2008.

\bibitem{dunik2009methods}
J.~Dun{\'\i}k, M.~{\^S}imandl, and O.~Straka, ``Methods for estimating state
  and measurement noise covariance matrices: Aspects and comparison,'' {\em
  IFAC Proceedings Volumes}, vol.~42, no.~10, pp.~372--377, 2009.

\bibitem{kalman1960general}
R.~E. Kalman, ``On the general theory of control systems,'' in {\em Proceedings
  First International Conference on Automatic Control, Moscow, USSR},
  pp.~481--492, 1960.

\bibitem{pearson1966duality}
J.~Pearson, ``On the duality between estimation and control,'' {\em SIAM
  Journal on Control}, vol.~4, no.~4, pp.~594--600, 1966.

\bibitem{bu2019lqr}
J.~Bu, A.~Mesbahi, M.~Fazel, and M.~Mesbahi, ``{LQR} through the lens of first
  order methods: Discrete-time case,'' {\em arXiv preprint arXiv:1907.08921},
  2019.

\bibitem{bu2020policy}
J.~Bu, A.~Mesbahi, and M.~Mesbahi, ``Policy gradient-based algorithms for
  continuous-time linear quadratic control,'' {\em arXiv preprint
  arXiv:2006.09178}, 2020.

\bibitem{fazel2018global}
M.~Fazel, R.~Ge, S.~Kakade, and M.~Mesbahi, ``Global convergence of policy
  gradient methods for the linear quadratic regulator,'' in {\em Proceedings of
  the 35th International Conference on Machine Learning}, vol.~80,
  pp.~1467--1476, PMLR, 2018.

\bibitem{fatkhullin2020optimizing}
I.~Fatkhullin and B.~Polyak, ``Optimizing static linear feedback: Gradient
  method,'' {\em SIAM Journal on Control and Optimization}, vol.~59, no.~5,
  pp.~3887--3911, 2021.

\bibitem{krasiler2024output}
S.~Kraisler and M.~Mesbahi, ``Output-feedback synthesis orbit geometry:
  Quotient manifolds and lqg direct policy optimization,'' {\em IEEE Control
  Systems Letters}, vol.~8, pp.~1577--1582, 2024.

\bibitem{mohammadi2021linear}
H.~Mohammadi, M.~Soltanolkotabi, and M.~R. Jovanovic, ``On the linear
  convergence of random search for discrete-time {LQR},'' {\em IEEE Control
  Systems Letters}, vol.~5, no.~3, pp.~989--994, 2021.

\bibitem{zhao2021global}
F.~Zhao, K.~You, and T.~Ba{ş}ar, ``Global convergence of policy gradient
  primal-dual methods for risk-constrained {LQR}s,'' {\em arXiv preprint
  arXiv:2104.04901}, 2021.

\bibitem{talebi2024ergodic}
S.~Talebi and N.~Li, ``Ergodic-risk criterion for stochastically stabilizing
  policy optimization,'' {\em arXiv preprint arXiv:2409.10767}, 2024.

\bibitem{Tang2021analysis}
Y.~Tang, Y.~Zheng, and N.~Li, ``Analysis of the optimization landscape of
  linear quadratic gaussian ({LQG}) control,'' in {\em Proceedings of the 3rd
  Conference on Learning for Dynamics and Control}, vol.~144, pp.~599--610,
  PMLR, June 2021.

\bibitem{talebi2022policy}
S.~Talebi and M.~Mesbahi, ``Policy optimization over submanifolds for
  constrained feedback synthesis,'' {\em IEEE Transactions on Automatic Control
  (to appear), arXiv preprint arXiv:2201.11157}, 2022.

\bibitem{talebi2024policy}
S.~Talebi, Y.~Zheng, S.~Kraisler, N.~Li, and M.~Mesbahi, ``Policy optimization
  in control: Geometry and algorithmic implications,'' {\em arXiv preprint
  arXiv:2406.04243}, 2024.

\bibitem{talebi2022duality}
S.~Talebi, A.~Taghvaei, and M.~Mesbahi, ``Duality-based stochastic policy
  optimization for estimation with unknown noise covariances,'' {\em arXiv
  preprint arXiv:2210.14878}, 2022.

\bibitem{talebi2023data}
S.~Talebi, A.~Taghvaei, and M.~Mesbahi, ``Data-driven optimal filtering for
  linear systems with unknown noise covariances,'' in {\em Advances in Neural
  Information Processing Systems ({NeurIPS})}, vol.~36, pp.~69546--69585,
  Curran Associates, Inc., 2023.

\bibitem{belabbas2025interpretable}
M.~A. Belabbas and A.~Olshevsky, ``Interpretable gradient descent for kalman
  gain,'' {\em arXiv preprint arXiv:2507.14354}, 2025.

\bibitem{Khlebnikov2023comparison}
M.~V. Khlebnikov, ``A comparison of guaranteeing and kalman filters,'' {\em
  Automation and Remote Control}, vol.~84, pp.~389--411, 2023.

\bibitem{talebi_riemannian_2022}
S.~Talebi and M.~Mesbahi, ``Riemannian {Constrained} {Policy} {Optimization}
  via {Geometric} {Stability} {Certificates},'' in {\em 2022 {IEEE} 61st
  {Conference} on {Decision} and {Control} ({CDC})}, pp.~1472--1478, 2022.

\bibitem{kwakernaak1969linear}
H.~Kwakernaak and R.~Sivan, {\em Linear Optimal Control Systems}, vol.~1072.
\newblock Wiley-interscience, 1969.

\bibitem{Tse1970optimal}
E.~Tse and M.~Athans, ``Optimal minimal-order observer-estimators for discrete
  linear time-varying systems,'' {\em IEEE Transactions on Automatic Control},
  vol.~15, no.~4, pp.~416--426, 1970.

\bibitem{lewis1986optimal}
F.~Lewis, {\em Optimal Estimation with an Introduction to Stochastic Control
  Theory}.
\newblock New York, Wiley-Interscience, 1986.

\bibitem{gajic2008lyapunov}
Z.~Gajic and M.~T.~J. Qureshi, {\em Lyapunov Matrix Equation in System
  Stability and Control}.
\newblock Courier Corporation, 2008.

\bibitem{talebi2026riemannian}
S.~Talebi and L.~Bier, ``Riemannian-regularized-policy-optimization,'' Mar.
  2026.
\newblock Available on GitHub at
  \url{https://github.com/shahriarta/Riemannian-regularized-policy-optimization}.

\end{thebibliography}

\end{document}